\documentclass{article}

\usepackage{amsmath,mathtools,amsthm,ifmath,ifart,ifmisc,palatino,graphicx,caption,array,enumitem,booktabs}
\usepackage[slantedGreek]{mathpazo}
\usepackage{MnSymbol}
\usepackage{accents}

\usepackage{xcolor}
\usepackage[colorlinks=true]{hyperref}
\definecolor{citecolor}{rgb}{0.1,0,0.5}
\definecolor{linkcolor}{rgb}{0,0.1,0.5}
\hypersetup{
	citecolor=citecolor,
	linkcolor=linkcolor
}

\usepackage{macros}

\begin{document}

\iftitle{Limit shapes for Gibbs ensembles of partitions}
\ifauthor{Ibrahim Fatkullin}{University of Arizona}
\ifauthor{Valeriy Slastikov}{University of Bristol}

\ifabstract{We explicitly compute limit shapes for several grand canonical Gibbs ensembles of partitions of integers. These ensembles appear in models of aggregation and are also related to invariant measures of zero range and coagulation-fragmentation processes.  We show, that all possible limit shapes for these ensembles fall into several distinct classes determined by the asymptotics of the internal energies of aggregates.}

\section{Introduction} 
Imagine a fog, a colloid, or a polymeric melt: systems in which identical primitives form aggregates of various sizes. In this paper we study the size distributions of such aggregates by analyzing limit shapes for various Gibbs-type probability measures on partitions.  These measures are prescribed by specifying the internal energy of aggregates as a function of their size. The asymptotic behavior of internal energy at infinity then determines which particular limit shape is chosen (or whether it exists altogether). In what follows, we occasionally borrow from the polymer physics language and call the primitives {\em monomers} and the aggregates --- {\em polymers.} Similar measures also appear as invariant measures of zero range processes in models of Bose-Einstein condensation \cite{andjel1982invariant, pitman1875combinatorial, funaki2010hydrodynamic, ercolani2014random}; in genetics \cite{ewens1972sampling,kingman1978random}; and in the context of general coagulation-fragmentation phenomena \cite{bertoin2006random}.

Let us first review some definitions and fundamental results related to partitions which are utilized in this work. Take a positive integer, $M$ and represent it as a sum of positive integers; this representation is called a {\em partition,} e.g., here is a partition of the number 14: 
\begin{equation}
	14=1+1+2+2+3+5.
\end{equation}
The number $M$ may be regarded as the number of monomers, and each summand in the partition --- as a polymer. Given a partition, denote the number of summands equal to $k$ by $\pn_k$; use the bold-face $\vv\pn$ to denote the entire sequence $(\pn_k)$, $k\in\N$. For example, the partition above corresponds to $\vv\pn=(2,2,1,0,1,0,0,\ldots)$. Any given \smash{$\Z^+$}-valued sequence (a sequence of non-negative integers) $\vv\pn$ for which 
\begin{equation}\label{eq:def_mass}
	\Mon(\vv\pn)\,\ass\,\sum_{k=1}^{\infty}k\pn_k
\end{equation}
is finite is a partition of $M=\Mon(\vv\pn)$; we refer to $\Mon(\vv\pn)$ as the {\em size} or the {\em mass} of the sequence $\vv\pn$. We say that $\vv\pn$ {\em partitions} $M$ and denote it by $\vv\pn\prn M$. $\Mon(\vv\pn)$ may be regarded as the total mass of a polymeric system whose state is described by $\vv\pn$. We denote the set of all partitions of $M$ by $\setP_M$, and the union of sets of all partitions of all positive integers by $\setP$,
\begin{equation}
	\setP=\bigcup_{M=0}^\infty\setP_M.
\end{equation}
For the purpose of convenience, we agree that $M=0$ has a single partition, $\vv\pn=(0,0,0,\ldots)$.
The number of summands greater than or equal to $x$ in a given partition, $\vv\pn$, is called  the {\em size distribution function} of the partition,
\begin{equation}\label{eq:dist_1d}
	f(x;\vv\pn)\,\ass\,
	\sum_{k\geq x}\pn_k.
\end{equation}
Observe that
\begin{equation}
	\int_0^\infty f(x;\vv\pn)\md x
	\,=\,\Mon(\vv\pn).
\end{equation}
There is a straightforward correspondence between partitions of integers, Young diagrams, and distribution functions which is illustrated in Figure~\ref{fig:mono_poly}.  In what follows we identify partitions of integers, Young diagrams and sequences of nonnegative integers with finite mass, and use elements of $\setP$ to represent all of these objects.

%
\begin{figure}
	\hfill
	\scalebox{0.6}{\includegraphics{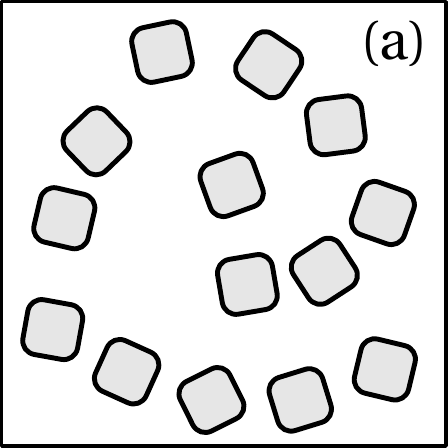}}\hfill
	\scalebox{0.6}{\includegraphics{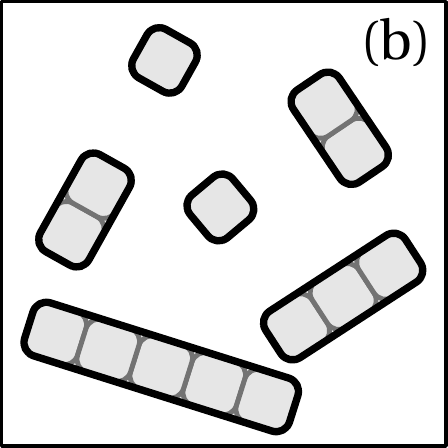}}\hfill
	\scalebox{0.6}{\includegraphics{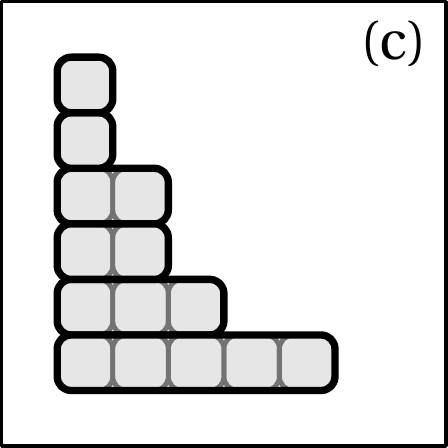}}\hfill
	\scalebox{0.6}{\includegraphics{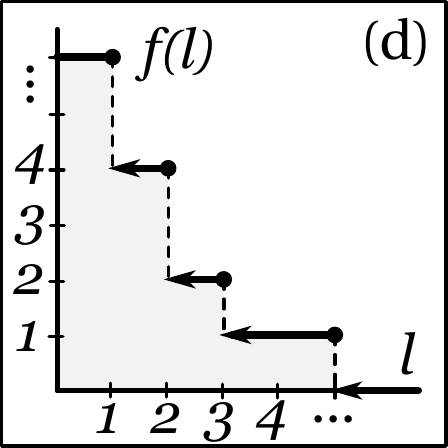}}\hfill\rule{0pt}{0pt}
	\caption{\label{fig:mono_poly}
		$14$ {\em monomers} (a) form $6$ {\em polymers} (b) according to the partition, $14=1+1+2+2+3+5$. Stacking these polymers on top of each other in non-increasing order, we get a Young (Ferrers) diagram (c), which itself is the {\em subgraph} of the distribution funtion $f(x)$ defined in equation \eqref{eq:dist_1d}.}
\end{figure}

\paragraph{Uniform measures on partitions of integers.} Fix $M$, let all partitions of $M$ be equiprobable, e.g., the partition, $14=1+1+2+2+3+5$, has the same probability, $1/Q_{14}$, as $14=3+5+6$. Here $Q_M$ denotes the {\em partition number} of $M$, i.e., the number of ways to represent $M$ as a sum of positive integers. Recall the Hardy-Ramanujan asymptotic formula, \cite{Hardy18} 
\begin{equation}\label{eq:HR}
	Q_M
	\,\sim\,
	\frac{1}{4M\sqrt{3}}
	\exp\{\pi\sqrt{2M/3}\},
	\quad\text{as}\quad
	M\to\infty.
\end{equation}

Once we prescribe a probability measure on $\setP_M$ (or $\setP$), the distribution function becomes a random function. What is the {\em typical} behavior of $f(x;\vv\pn)$ as $M\to\infty$? As $M$ grows, a typical summand in its partition grows as well, so we must apply a proper scaling to $f(x;\vv\pn)$ in order to obtain a nontrivial limit. Consider
\begin{equation}\label{eq:size_dist_scaled}
	F_M(x;\vv\pn)\,=\,f(x\sqrt{M};\vv\pn)/\sqrt{M}.
\end{equation}
This scaling effectively shrinks the size of squares in a Young diagram by a factor of $\sqrt{M}$; see Figure~\ref{fig:lim_shape}. Note  that the total integral of $F_M(x)$ over $\R^+$ is unity whenever $\vv\pn\prn M$.  It turns out that as $M\to\infty$, with overwhelming probability, $F_M(x;\vv\pn)$ concentrate near a deterministic {\em limit shape},
\begin{equation}\label{eq:lim_shape_form}
	F(x)\,=\,-\frac{\sqrt{6}}{\pi}\ln\Big(1-\me^{-\pi x/\!\sqrt{6}}\Big).
\end{equation}
Namely, for all $a,b$, such that $0<a<b<\infty$; $\epsilon>0$, there exists $M_0$ such that for all $M>M_0$,
\begin{equation}\label{eq:lim_shape1}
	\Prob_{M}\bigg\{\sup_{x\in[a,b]}|F_M(x;\vv\pn)-F(x)|>\epsilon\bigg\}\,<\,\epsilon.
\end{equation}
Here $\Prob_{M}$ denotes the uniform probability measure on partitions of $M$. This result was proven by A.~Vershik, who also proposed a general method for analysis of limit shape problems in a number of related systems \cite{vershik1996statistical}; his method is utilized extensively in this work. The limit shape formula \eqref{eq:lim_shape_form} first appeared in a study by M.~Szalay and R.~Turan \cite{szalay1977some}. 

%
\begin{figure}[b]
	\hfill\qquad
	\scalebox{0.25}{\includegraphics{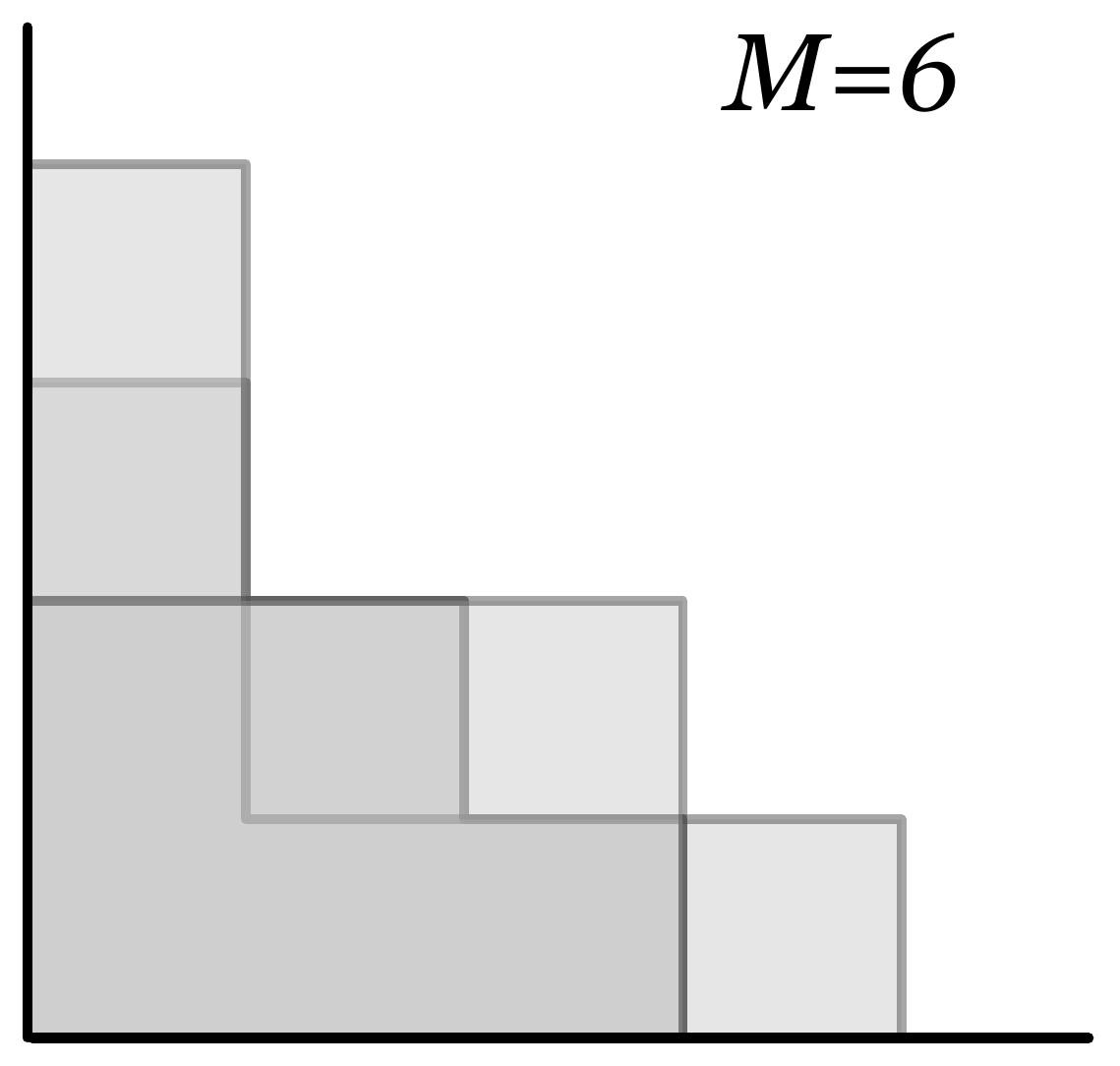}}\hfill
	\scalebox{0.25}{\includegraphics{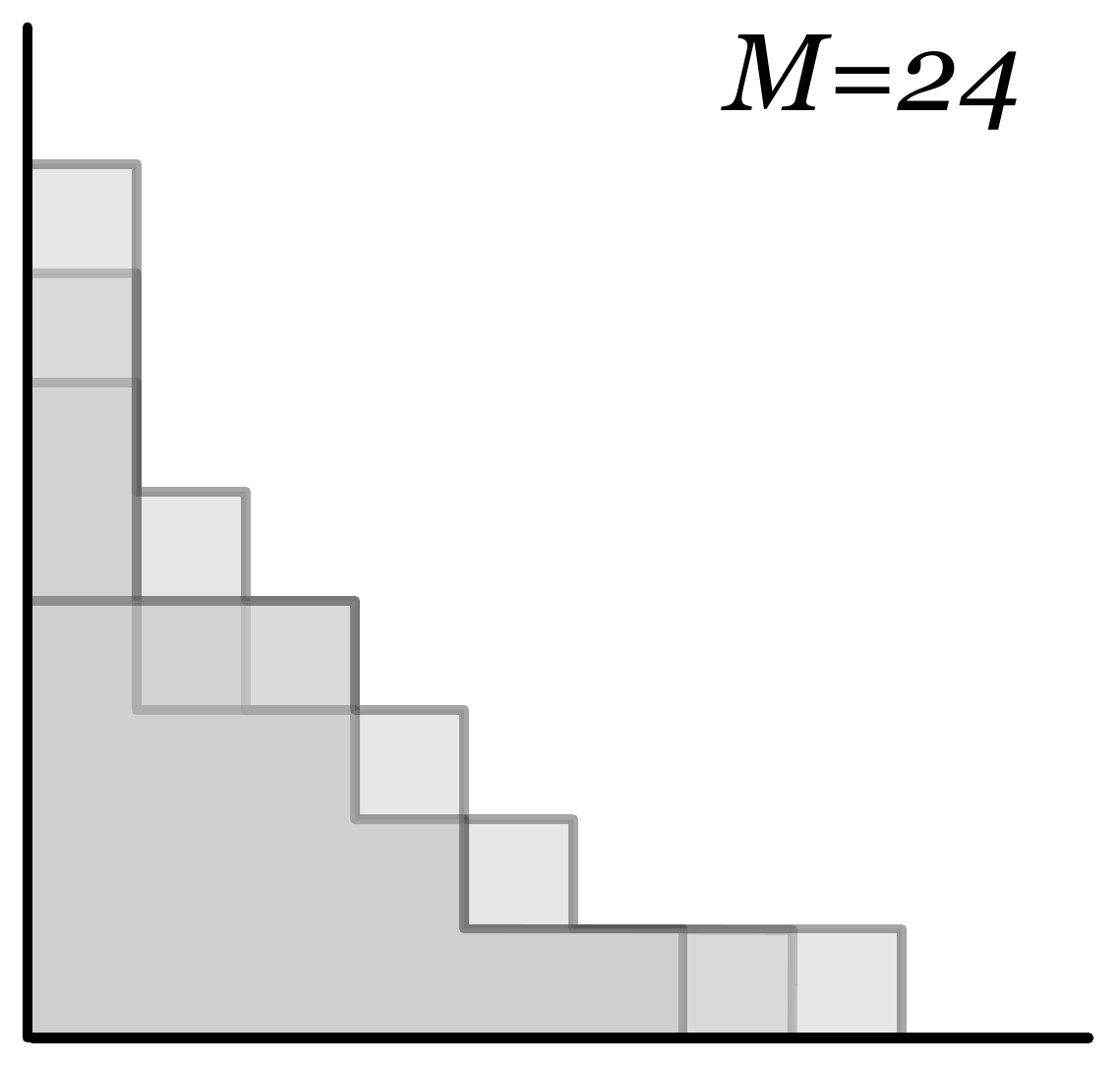}}\hfill
	\scalebox{0.25}{\includegraphics{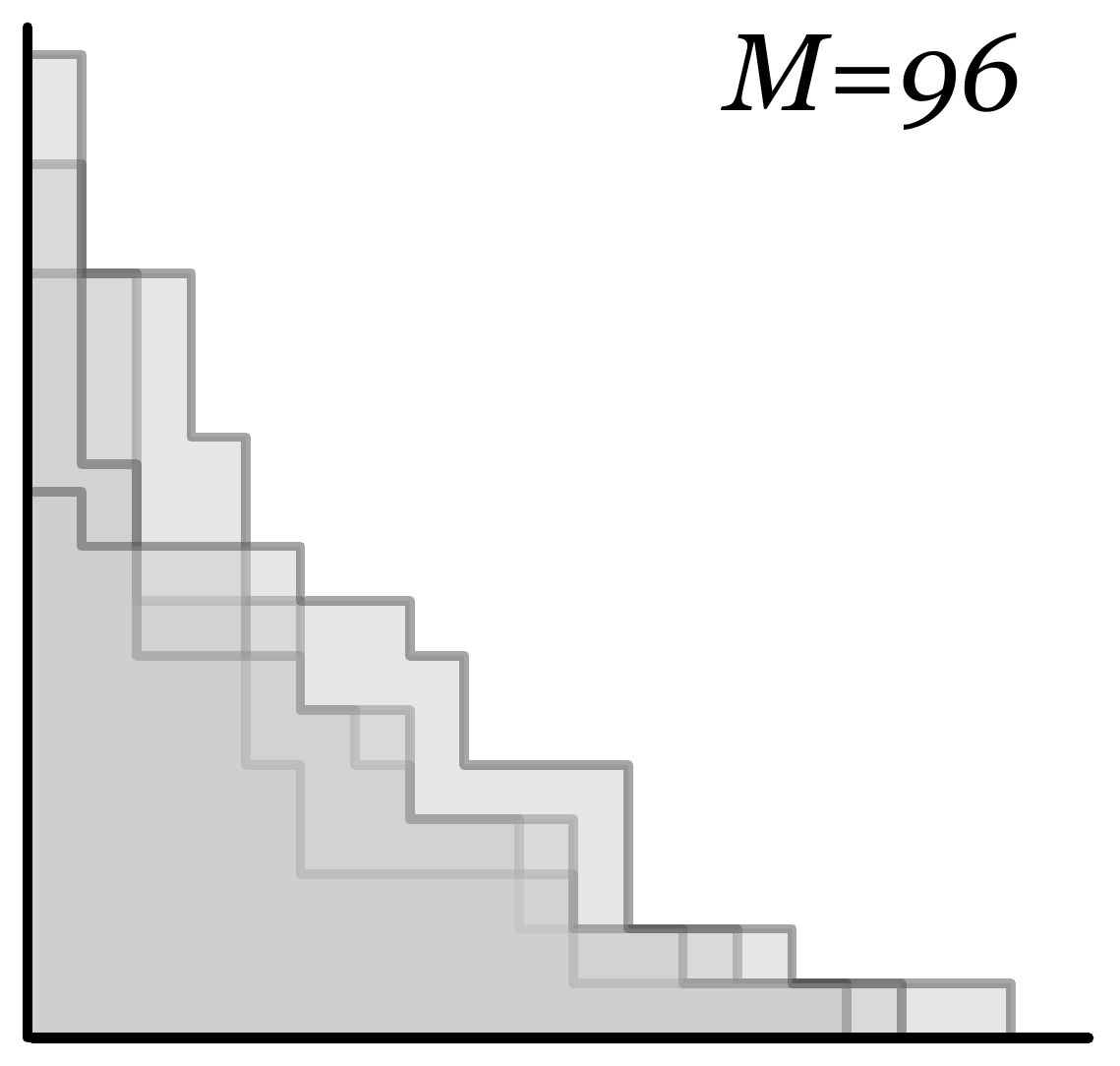}}\hfill
	\scalebox{0.25}{\includegraphics{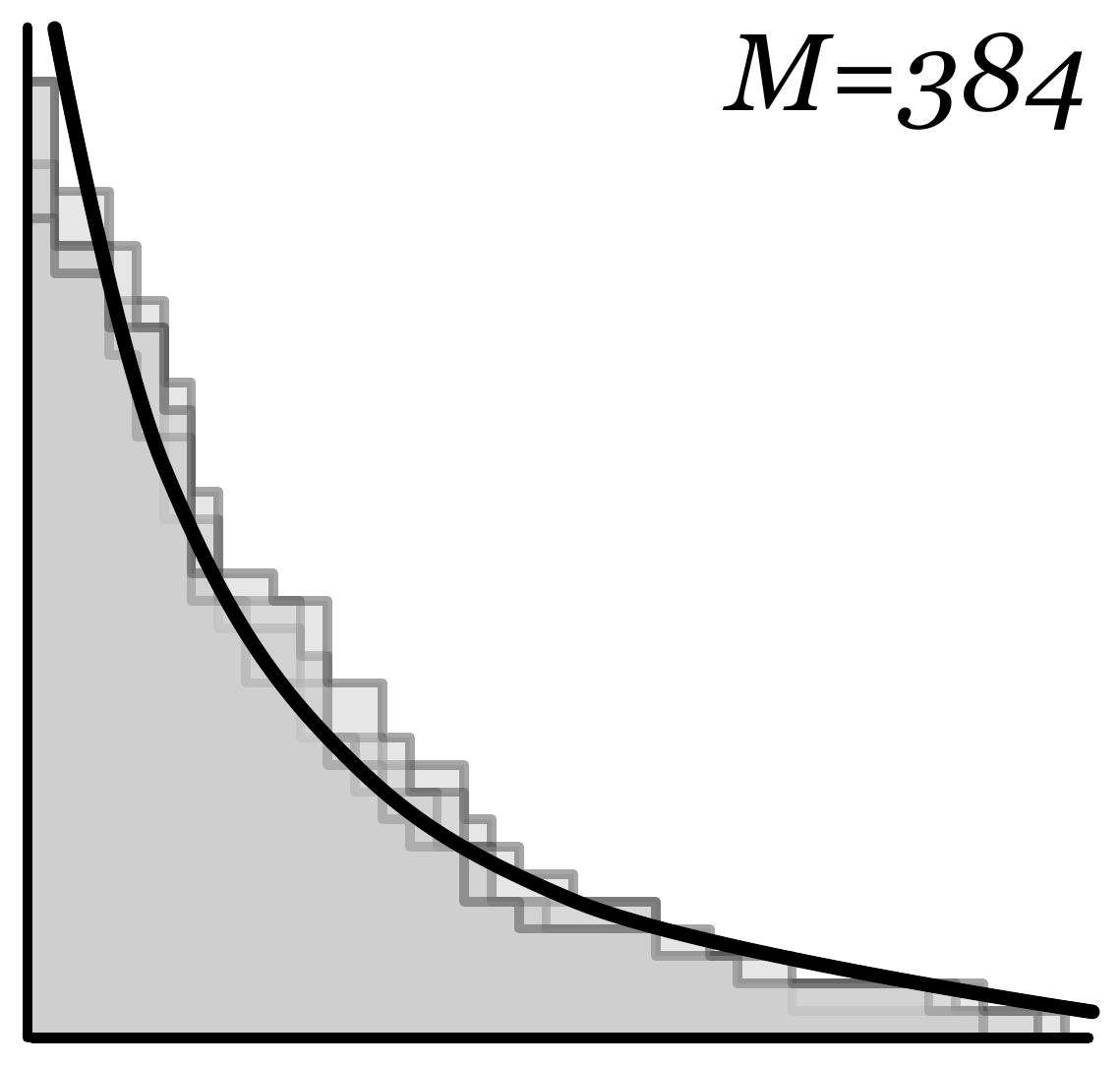}}\qquad\hfill\rule{0pt}{0pt}
	\caption{\label{fig:lim_shape}
		A few sample functions, $F_M(x;\vv\pn)$, are depicted in grey color for several values of $M$. As $M$ becomes large, their {\em typical} graphs approach the deterministic {\em limit shape,} $F(x)$ --- black line in the right-most plot.}
\end{figure}

In the language of polymers, the statement above means that in a system where all groupings of monomers into polymers are equiprobable, in the {\em thermodynamic limit} as $M\to\infty$, there exists a limiting polymer size density given (after appropriate rescaling) by
\begin{equation}
	\od(x)\,=\,-F^\prime(x)\,=\,
	\frac{1}{\me^{\,\pi x/\!\sqrt{6}}-1}.
\end{equation}
Note that in this model the monomers are {\em indistinguishable,} e.g., in a partition, $3=1+2$, we do not differentiate which of the two monomers formed the polymer of size two and which one remained separate. Because of this, we sometimes refer to such systems as {\em quantum}, alluding to the indistinguishability of quantum particles.

\paragraph{Uniform measures on partitions of sets: Bell statistics.} If, unlike in the model described above, we want to distinguish the monomers, then we must consider partitions of sets rather than integers. For example, the partition $3=1+2$ corresponds to three distinct set partitions: $\{\{m_1\},\{m_2,m_3\}\}$, $\{\{m_2\},\{m_1,m_3\}\}$, and $\{\{m_3\},\{m_1,m_2\}\}$, each specifying which particular monomers, elements of the set $\{m_1,m_2,m_3\}$, formed a polymer of size two and which one remained unattached. Because the monomers in this model are {\em distinguishable,} as are classical particles, we refer to such systems as {\em classical}. The name {\em Bell statistics} honors E.~T.~Bell who studied partitions of sets in the first half of the twentieth century.

Set partitions may also be characterized by sequences in $\setP$: a particular sequence $\vv\pn$ corresponds to all partitions of a set with cardinality $M=\Mon(\vv\pn)$, which contain exactly $\pn_k$ subsets of cardinality $k$. There are $M!/(\prod_{k=1}^{\infty}(k!)^{\pn_k}\pn_k!)$
such set partitions (although this product is infinite, only a finite number of terms in it differs from one). This implies that the uniform probability measures on partitions of sets induce the following measures on $\setP_M$:
\begin{equation}\label{eq:bell_stat}
	\Prob_M\{\vv\pn=\vv P\}
	\,=\,
	\frac{M!}{B_M}\prod_{k=1}^\infty\frac{1}{(k!)^{P_k}P_k!}.
\end{equation}
Here $B_M$ is the $M$-th Bell number --- the number of partitions of a set with cardinality $M$. 
The asymptotic behavior of these measures as $M\to\infty$ has been investigated by Yu.~Yakubovich \cite{yakubovich1995asymptotics}. The appropriate scaling for the limit shape is
\begin{equation}
	F_M(x;\vv\pn)\,=\,\me^{-\nu} f(\nu x;\vv\pn),
	\qquad\text{where~}\nu\text{~solves}\quad
	\nu\me^\nu\,=\,M.
\end{equation}
The limit shape itself is given by
\begin{equation}\label{eq:step}
	F(x)\,=\,\ind{\{x\leq1\}}(x)
\end{equation}
where 
$\ind{\set A}(\cdot)$ denotes the indicator function of a set $\set A$. This implies that the size density of aggregates becomes atomic, i.e., in the thermodynamic limit, after the appropriate rescaling, all polymers have size one.

\paragraph{A few more examples: Haar statistics; Plancherel measure.} Another measure on partitions may be obtained by pushing forward the uniform (Haar) measure on the symmetric group. Each permutation of $M$ elements (of some set) may be decomposed into cycles; denoting by $\pn_k$ the number of cycles of length $k$, we get a natural surjection of the symmetric group onto $\setP_M$. The resulting measure is prescribed (up to normalization) by a formula similar to that in \eqref{eq:bell_stat}, except in the denominator one has $k^{P_k}$ instead of \smash{$(k!)^{P_k}$}. In terms of polymers, this model describes a {\em classical} ensemble of {\em loop polymers,} i.e., polymers whose ends are connected to form a loop. Indeed, there are exactly $(k-1)!$ ways to arrange $k$ distinguishable monomers into a loop, which explains the $k^{P_k}$ factor. These measures, however, do not have a nontrivial limit shape for any rescaling of the size distribution function \cite{schmidt1977limit}.

Another celebrated family of measures on partitions of integers (or Young diagrams) is that of {\em Plancherel} measures. These measures are related to irreducible representations of the symmetric group of $M$ elements, which may be parametrized by Young diagrams of size $M$. The Plancherel measure on $\setP_M$ is given by 
\begin{equation}
	\Prob_M\{\vv\pn=\vv P\}
	\,=\,\frac{\dim^2(\vv P)}{M!}.
\end{equation}
Here $\dim(\vv P)$ is the dimension of the irreducible representation labeled by $\vv P$; it is equal to the number of {\em Young tableaux} (diagrams filled with numbers, non-increasing along rows and columns) corresponding to the diagram $\vv P$. The limit shape problem for the Plancherel measures was first studied by B.~F.~Logan and L.~A.~Shepp \cite{logan1977variational}, and by S.~Kerov and  A.~Vershik \cite{kerov1977asymptotics}. The corresponding limit shape function does not have a simple explicit form, however, it is described by the following formula if the coordinate frame is rotated clockwise by $45^\circ$:
\begin{equation}
	y(x)\,=\,\frac{2}{\pi}\left(x\arcsin(x/2)\,+\,\sqrt{4-x^2}\right),\quad |x|\leq2.
\end{equation}
For more details, we refer the reader to the aforementioned studies, or to the subsequent works by A.~Borodin, A.~Okounkov, and G.~Olshanski, e.g., \cite{borodin2000asymptotics}, who refined the earlier results, studied fluctuations near the limit shape and related them to determinantal point processes.
\subsection{Gibbs ensembles of integer partitions}\label{sec:Gint}
The examples reviewed above are the most fundamental measures on partitions. One might say that they are purely {\em entropic}, i.e., they describe the configuration spaces of respective systems. In many physically interesting contexts, one must weigh these configuration spaces with some additional factors which account for interactions between the monomers (or polymers) in the system. In this paper we calculate the limit shapes for partitions of integers weighted by factors representing the internal energies of (quantum) polymers. The Gibbs ensembles are prescribed by specifying the energy (Hamiltonian) of partitions; in the simplest form, 
\begin{equation}\label{eq:ham1}
	\Ham(\vv\pn)
	\,\ass\,
	\sum_{k=1}^{\infty}E_k \pn_k.
\end{equation}
The numbers $E_k$ represent the {\em internal energies\,} of polymers of size $k$. The Hamiltonians prescribed by equation \eqref{eq:ham1} describe systems in which only the monomers within the same polymer interact with each other. 

\paragraph{A general pair-wise interaction.} Suppose each polymer is characterized by some {\em state} variable, $s$; it may represent, e.g., polymer's location in physical space, orientation, etc. Let $\pn_{k,s}$ denote the number of polymers of size $k$ in the state $s$. The polymers may also interact with each other; let the energy of interaction of two polymers of sizes $k$ and $k^\prime$ and states $s$ and $s^\prime$ be given by $\intr_{k,s;\,k^\prime\!,s^\prime}$. The total energy of such a system is then given by
\begin{equation}
	\Ham(\vv\pn)
	\,=\,
	\sum_{(k,s)}E_{k,s}\, \pn_{k,s} \;+
	\frac{1}{2}\sum_{(k,s)}\;\sum_{(k^\prime\!\!,\,s^\prime)\neq(k,s)}
	\intr_{k,s;\,k^\prime\!\!,\,s^\prime}\,\pn_{k,s}\pn_{k^\prime\!\!,\,s^\prime}.
\end{equation}
This scenario is quite general and encompasses all systems with pair-wise interactions; we do not consider such generality here and concentrate specifically on simpler Hamiltonians as given in \eqref{eq:ham1}.

\paragraph{Canonical ensembles} are probability measures defined on $\setP_M$-s --- partitions of a particular integer, $M$. Given a Hamiltonian $\Ham(\vv p)$ and the {\em inverse temperature} $\beta$, prescribe
\begin{equation}\label{eq:canon}
	\Prob_M\{\vv\pn=\vv P\}
	\,=\,
	\frac{1}{\cpsum{M}{\beta}}\,\me^{-\beta\Ham(\vv P)};\qquad
	\cpsum{M}{\beta}
	\,=\,
	\sum_{\vv\pn\prn M}\me^{-\beta\Ham(\vv P)}.
\end{equation}
When $\beta=0$, the {\em partition sum} $\cpsum{M}{\beta}$ is equal to the partition number of $M$, $Q_M$, i.e., the inverse temperature controls the strength of interactions in the system. One can also treat the canonical measures as measures on the set of all partitions, $\setP$. In this case, the size of the partition, $\Mon(\vv p)$, becomes a random function and $\Prob_M\{\vv\pn=\vv P\}=0$ whenever $\Mon(\vv p)\neq M$. Let us mention that the problem of establishing the asymptotic behavior of the partition sums $\cpsum{M}{\beta}$ given the energies $E_k$, i.e., obtaining the generalizations of the Hardy-Ramanujan formula \eqref{eq:HR}, is an important problem in the field of enumerative combinatorics, see \cite{Gran2015} for some recent developments.

\paragraph{The grand canonical ensembles}
are defined on the set of all partitions, $\setP$. They involve yet another parameter, $\mu$, which we call the {\em chemical potential} it controls the expected number of monomers in the system.  The grand canonical measures are prescribed by setting
\begin{equation}\label{eq:gcanon}
	\Prob_\mu\{\vv \pn=\vv P\}
	\,=\,
	\frac{1}{\gpsum{\mu}{\beta}}\me^{-\beta\Ham(\vv P)-\mu\Mon(\vv P)}.
\end{equation}
Note that we employed a nonstandard definition of chemical potential, which normally appears in the combination, $\exp\{-\beta(\Ham+\mu\Mon)\}$. This helps us reduce clutter in various formulas and makes little difference otherwise (as long as $\beta$ is not sent to zero or infinity).

The grand canonical partition functions, $\gpsum{\mu}{\beta}$, may be expressed as
\begin{equation}\label{eq:gc_part_sum}
	\gpsum{\mu}{\beta}
	\,=\,
	\sum_{\vv\pn\in\setP}\me^{-\beta\Ham(\vv\pn)-\mu\Mon(\vv\pn)}
	\,=\,
	\sum_{M=0}^{\infty}\sum_{\vv\pn\prn M}\me^{-\beta\Ham(\vv\pn)-\mu M}
	\,=\,
	\sum_{M=0}^{\infty}\cpsum{M}{\beta}\,\me^{-\mu M}.
\end{equation}
The last equality implies that the grand canonical partition function, $\gpsum{\mu}{\beta}$ (when treated as a function of $\me^{-\mu}$) is the generating function for the canonical partition sums, $\cpsum{M}{\beta}$. 
Correspondingly, the canonical measures are the conditional restrictions of the grand canonical measures to $\setP_M$-s:
\begin{equation}
	\Prob_M(\cdots)\,=\,\Prob_\mu(\cdots\mid\Mon(\vv p)=M);\qquad
	\Prob_\mu
	\,=\,
	\frac{1}{\gpsum{\mu}{\beta}}
	\sum_{M=0}^\infty
	\cpsum{M}{\beta}\,\me^{-\mu M}\,\Prob_M.
\end{equation}
A fundamental observation made by A. Vershik \cite{vershik1996statistical} is that the grand canonical measures are multiplicative: the partition functions may be represented as a product,
\begin{equation}\label{eq:part_sum_prod}
	\gpsum{\mu}{\beta}
	\,=\,
	\sum_{m=0}^{\infty}\sum_{\vv\pn\prn m}\prod_{k=1}^{\infty}
	\me^{-(\beta E_k+\mu k)\pn_k}
	\,=\,
	\prod_{k=1}^{\infty}\sum_{\pn_k=\,0}^\infty\left(\me^{-\beta E_k-\mu k}\right)^{\!\pn_k}
	\,=\,
	\prod_{k=1}^\infty\frac{1}{1-\me^{-\beta E_k-\mu k}};
\end{equation}
and, respectively,
\begin{equation}\label{eq:independent}
	\Prob_\mu\{\vv \pn=\vv P\}
	\,=\,
	\prod_{k=1}^\infty\Prob_\mu^{(k)}\{p_k=P_k\}
	;\qquad
	\Prob_\mu^{(k)}\{p_k=N\}
	\,=\,{\theta_k^{N}}(1-\theta_k),
	\quad N=0,1,2,\ldots
\end{equation}
Here we denoted
\begin{equation}\label{eq:thetas}
	\theta_k\,\ass\,
	\me^{-\beta E_k-\mu k}.
\end{equation}
Thus the multiplicative property implies that in the grand canonical ensembles (unlike in the canonical ensembles) the numbers of polymers of different sizes are independent random variables, which simplifies analysis of such systems.

\paragraph{Equivalence of ensembles.}
In order to study the limit shape problem, we must consider measures induced by $\Prob_M$ or $\Prob_\mu$ on some suitable function space via push-forward maps by the properly scaled size distribution functions, $f(x;\vv\pn)$, defined in equation \eqref{eq:dist_1d}. The equivalence of ensembles means that in the {\em thermodynamic} limit as $M\to\infty$ and $\mu\dto\mu_*$ (for a suitable $\mu_*$, see below) the measures induced by the canonical and grand canonical ensembles converge to the same limit. This has been proven for multiplicative measures in cases when the limit distributions are atomic, i.e., concentrated on a single function --- the limit shape \cite{vershik1996statistical,vershik1997limit}. In the context of our work, we only study the grand canonical ensembles; the corresponding results regarding the large $M$ limits in the canonical setting follow whenever the limit shape exists. 
\paragraph{Internal energies.} Finally, let us discuss the classes of internal energies, $E_k$, for which the grand canonical ensembles are well-defined and correspond to physically relevant scenarios. From the basic theory of convergence of infinite series and products we deduce that the necessary condition for existence of the partition functions in equations \eqref{eq:gc_part_sum} or \eqref{eq:part_sum_prod} and the grand canonical measures, is that $\mu k\geq-\beta E_k$ for all $k$. Introduce the energies of monomers inside of a polymer of size $k$,
\begin{equation}\label{eq:prt_fun_fin}
	\eps_k\,=\,\frac{E_k}{k};
	\qquad\text{denote}\qquad
	\eps_*
	\,\ass\,
	\inf_{k\in\N}\eps_k
	\,=\,
	-\sup_{k\in\N}(-\eps_k).
\end{equation}
We refer to $\eps_*$ as the {\em ground state} energy, even though it is not necessarily attained for any particular value of the polymer size, $k$. The grand canonical measures exist for all $\mu$ which satisfy
\begin{equation}\label{eq:mustar}
	\mu\,>\,\mu_*\,\ass\,-\beta\eps_*.
\end{equation}
(The case $\mu=\mu_*$ requires additional considerations.) This observation leads to the following scenarios:
\begin{enumerate}[label=(S\arabic*)]
	\item\label{case1} The infimum in equation \eqref{eq:prt_fun_fin} is $\eps_*=-\infty$, attained in the limit as $k\to\infty$. In this case the grand partition function, $\gpsum{\mu}{\beta}$, is infinite for all $\mu$, i.e., the sums and products in equation \eqref{eq:part_sum_prod} are undefined, and the grand canonical measures do not exist for any value of the chemical potential, $\mu$. In order to analyze phenomena in the thermodynamic limit, one must work with canonical ensembles directly; we do not study this scenario here.
	\item\label{case2} The infimum is reached at some finite values of $k$. In this case $\gpsum{\mu}{\beta}\to\infty$ as $\mu\dto\mu_*$; the probability concentrates on the corresponding discrete set of states, $p_k$, and thus there is no limit shape. We call this phenomenon {\em condensation}. One can ``remove'' these states from the system and analyze the distribution on the remaining states, which again reduces to either this case, or the case \ref{case3} below. An example demonstrating this procedure is given in Section~\ref{sec:discrete_example}.
	\item\label{case3} The infimum is attained as $k$ tends to infinity, and $\eps_*$ is finite; we say that the ground state is {\em unattainable.} Without loss of generality, we may set $\eps_*=0$ and consider the limit $\mu\dto0$ (instead of $\mu\dto-\beta\eps_*$), as subtracting $\eps_*$ from each $\eps_k$ is equivalent to adding $\beta\eps_*$ to the chemical potential $\mu$: all probabilities depend on the sums $\beta\eps_k+\mu$. This paper is mostly dedicated to this scenario.
\end{enumerate}
%
%
%
\begin{figure}
	\hfill
	\scalebox{0.6}{\includegraphics{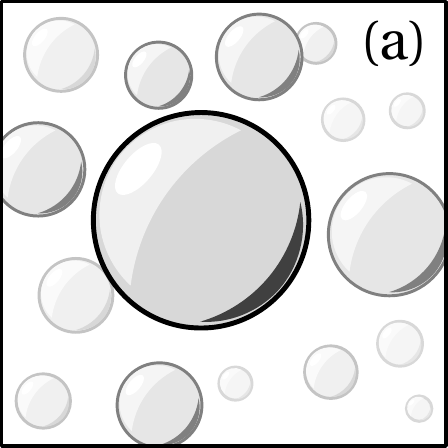}}\hfill
	\scalebox{0.6}{\includegraphics{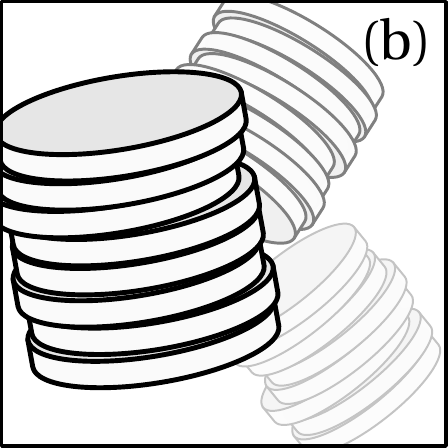}}\hfill
	\parbox[b]{3in}{\caption{\label{fig:energies}
		The surface area of a $d$-dimensional droplet (a) is proportional to its mass to power $(d-1)/d$. Each stack of disc-like molecules in discotic liquid crystals has roughly the same energy, independent of its size, which comes from the two hydrophobic faces on its ends (b). This simplistic reasoning justifies the energies suggested in formula \eqref{eq:energies}.\vspace{-2ex}}}\hfill\rule{0pt}{0pt}
\end{figure}
We see that whenever the internal energies $E_k$ grow super-linearly as $k\to\infty$, we end up with scenarios \ref{case1} or \ref{case2}. Sub-linear growth, e.g., 
\begin{equation}\label{eq:energies}
	E_k\,\sim\,k^\alpha,\quad\alpha\leq 1;\qquad\qquad
	E_k\,\sim\,\ln k;\qquad\qquad
	E_k\,\sim\,\ln\ln k.
\end{equation}
allows for more interesting behaviors; see Figure~\ref{fig:energies} for illustration and some extra motivation. Later we show (Proposition~\ref{prop:therm}) that faster than logarithmic growth does not allow for a thermodynamic limit in the context of grand canonical ensembles (the expected number of monomers in the system remains finite). This implies that in order to study the limit shapes in such ensembles, it is enough to limit ourselves to energies satisfying the following assumption: 
\begin{assum}\label{ass:nrg_quantum}
The internal energies of polymers are given by
\begin{equation}\label{eq:energies_u}
	E_k\,=\,u(\ln k).
\end{equation}
where $u:[0,\infty)\to(0,\infty)$ is a (strictly positive) differentiable function, such that $\lim_{x\to\infty}u^\prime(x)$ exists and is finite (without loss of generality it may be set to either 0 or 1, as all other cases may be covered by adjusting the value of $\beta$ appropriately.)
\end{assum}
%


\subsection{Gibbs ensembles of set partitions}
One can introduce the Gibbs ensembles on partitions of sets in the same way as we introduced them for the partitions of integers. In this case, the probabilities acquire additional combinatorial factors as in equation \eqref{eq:bell_stat}:
\begin{equation}\label{eq:canon_sets}
	\Prob_M\{\vv \pn=\vv P\}
	\,=\,
	\frac{1}{\cpsum{M}{\beta}}\me^{-\beta\Ham(\vv P)}\,
	\prod_{k=1}^\infty\frac{1}{(k!)^{P_k}P_k!};\qquad
	\cpsum{M}{\beta}
	\,=\,
	\sum_{\vv\pn\prn M}\frac{\me^{-\beta\Ham(\vv P)}}{(k!)^{P_k}P_k!}.
\end{equation}
Some authors refer to {\em these} ensembles as Gibbs measures on partitions, see e.g., the treatise by J.~Pitman \cite{pitman1875combinatorial}. Thus it is important to be aware of the difference between these measures and measures on the partitions of integers introduced here in Section~\ref{sec:Gint}, namely of the factors $P_k!$ in formula \eqref{eq:canon_sets}. Because of these factors, the distributions of individual $\pn_k$-s in the grand canonical ensembles of set partitions are Poisson, whereas the distributions of $\pn_k$-s for integer partitions are geometric, cf \eqref{eq:independent}.

Probabilities of particular partitions for the grand canonical ensembles are prescribed by
\begin{equation}\label{eq:gcanon_sets}
	\Prob_\mu\{\vv \pn=\vv P\}
	\,=\,
	\frac{1}{\gpsum{\mu}{\beta}}\me^{-\beta\Ham(\vv P)-\mu\Mon(\vv P)}\,
	\prod_{k=1}^\infty\frac{1}{(k!)^{P_k}P_k!};\qquad
	\gpsum{\mu}{\beta}
	\,=\,
	\sum_{M=0}^{\infty}\cpsum{M}{\beta}\,\me^{-\mu M}.
\end{equation}
These measures are also multiplicative:
\begin{equation}\label{eq:gibbs_bell_stat}
	\Prob_\mu\{\vv\pn=\vv P\}
	\,=\,
	\prod_{k=1}^\infty\Prob_\mu^{(k)}\{p_k=P_k\};\qquad
	\Prob_\mu^{(k)}\{p_k=N\}\,=\,\me^{-\alpha_k}\frac{\alpha_k^{N}}{N!},
	\qquad
	\alpha_k\,=\,
	\frac{\me^{-\beta E_k-\mu k}}{k!}. 
\end{equation}
Because the distributions of $p_k$-s are Poisson, rather than geometric, unlike in the {\em quantum} case, the {\em classical} grand canonical ensembles exist for arbitrary internal energies $E_k$ and for all $\mu$.

If the energies are sufficiently {\em tame,} i.e., if $E_k$ do not tend to negative infinity as fast as or faster than $-C\,k\ln k$, as \mbox{$k\to\infty$}, the thermodynamic limit is obtained when $\mu\to-\infty$. It is possible to show that for such energies the limit shape is given by the step function, as in \eqref{eq:step}, and the specific asymptotic behavior of the energies only affects the scaling for the size distribution function. The details of this calculation will be provided elsewhere. 

The so-called {\em expansive} case when \smash{$\me^{-\beta E_k}/k!\sim k^{p-1}$} was addressed by M.~Erlihson and B.~Granovsky directly in the canonical setting \cite{erlihson2008limit}, and also by A.~Cipriani and D.~Zeindler~\cite{cipriani2013limit} using different methods. Remarkably, the limit shape function that appears in these studies has the same functional form as one of those derived here in a different setting, see equation \eqref{eq:asympt_crit_shape} in the following section.


\subsection{Informal statement of the results}\label{sseq:inf_results}

Let us now state the principal results of this paper regarding the limit shapes for partitions of integers (quantum grand canonical ensembles) and interpret them in the polymer physics language. Results concerning the {\em classical} ensembles and partitions of sets will be discussed elsewhere. The more precise technical statements are formulated in the following Section~\ref{sec:calculations}.

Assume that the energies per monomer, $\eps_k$, are such that the scenario~\ref{case3} occurs, i.e., the ground state is unattainable; scenario~\ref{case2} is mentioned below and discussed in Section~\ref{sec:discrete_example} in greater detail. Let the energies $E_k$ be renormalized so that $\eps_*=0$, which may be achieved by changing $E_k\mapsto E_k-\eps_* k$ if $\eps_*\neq 0$. The following regimes are possible:


\paragraph{Supercritical growth, $\boldsymbol{E_k\gg\ln k}$.} 
The energies grow too fast at infinity, and the expected number of monomers, $\E\Mon(\vv p)$, in the grand canonical ensembles remains finite for all values of $\mu$. Because of this, one cannot obtain the thermodynamic limit in the grand canonical ensembles and must work with canonical ensembles, considering the limit of $M\to\infty$ directly. The limit $\mu\to0$ is trivial, obtained by setting $\mu=0$ in formulas \eqref{eq:independent} and \eqref{eq:thetas}. There is no limit shape, and the grand canonical distribution remains discrete (geometric distribution) with
\begin{equation}\label{eq:lim_discrete}
	\Prob\{p_k=N\}
	\,=\,{\me^{-\beta N E_k}}(1-\me^{-\beta E_k}).
\end{equation}
%


\paragraph{Critical growth, $\boldsymbol{E_k\sim\ln k}$.}
The value of the inverse temperature plays an important role in this case. For high temperatures, when $\beta<1$, the limit shape exists and, provided the width of cells in Young diagrams is set to $\mu$, while their height is set to \smash{$\mu^{1-\beta}/\Gfun(2-\beta)$}, is given by an incomplete $\Gfun$-function,
\begin{equation}\label{eq:asympt_crit_shape}
	F(x)
	\,=\,
	\frac{1}{\Gfun(2-\beta)}\int_x^\infty y^{-\beta}\me^{-y}\md y.
\end{equation}
For low temperatures, when $\beta>1$, there is no limit shape; physically this means that different realizations of the system do not have statistically similar length distributions. Technically, in this regime, the variance of the distribution functions tends to infinity as $\mu\dto0$ (see Section~\ref{sec:calculations}). The critical case, $\beta=1$, is more subtle and the limiting behavior depends on the higher-order asymptotic of $E_k$-s. For example, if we have the exact equality $E_k=\ln k$ for all or all large enough $k$, the scaled size distributions functions converge to a (non-stationary) stochastic process with independent increments, rather than to a deterministic limit shape function: see discussion in Section~\ref{sec:discrete_example} and Remark~\ref{rem:process}. There is also no automatic equivalence of grand canonical and canonical ensembles if $\beta=1$, as the limiting measure is not atomic.


\paragraph{Subcritical growth, $\boldsymbol{\const\ll E_k\ll\ln k}$.}
The limit shape in this regime is given by
\begin{equation}\label{eq:asympt_subcrit_shape} 
	F(x)
	\,=\,
	\me^{-x},
\end{equation}
provided that the height of cells in Young diagrams, is set to \smash{$\mu\exp(\beta E_{1/\mu})$} and the width is  set to $\mu$. Note that the limit shape itself does not depend on the value of the inverse temperature $\beta$: the latter only appears in the height to width scaling factor.


\paragraph{Constant, $\boldsymbol{E_k\sim1}$.} Provided the cell width is set to $\mu$ and the height to \smash{$\mu[\me^{\,\beta}\Li_2\big(\me^{-\beta}\big)]^{-1}$}, the limit shape is given by
\begin{equation}\label{eq:lim_shape_const}
	F(x)
	\,=\,
	\frac{1}{\Li_2\big(\me^{-\beta}\big)}
	\int_x^\infty\frac{\md y}{\me^{\,\beta+y}-\,1}
	\,=\,
	-\,\frac{\ln\Big(1-\me^{-\beta-x}\Big)}{\Li_2\big(\me^{-\beta}\big)}.
\end{equation}
Here $\Li_2(\cdot)$ denotes the dilogarithm. Observe that in this regime limit shape function is the same (up to the scaling factor of \smash{$\pi/\sqrt{6}$}) as the classical limit shape function \eqref{eq:lim_shape_form}, shifted by $\beta$; it also remains bounded  as $x\dto0$ as long as $\beta>0$. 


\paragraph{Decay, $\boldsymbol{E_k\ll\const}$.}
The limit shape (up to a trivial rescaling) is given by the same classical formula, \eqref{eq:lim_shape_form}, i.e., the internal energy does not affect the limit shape. This is not unexpected, as in this case, the energies of large polymers are too small and the entropic effects dominate.


\paragraph{Condensation: scenario \ref{case2}.}
In this scenario, infinitely many monomers form a {\em condensate,} i.e, aggregate into polymers with sizes which yield infimum in \eqref{eq:prt_fun_fin}. If infimum of the internal energies over the remaining states is greater than $\eps_*$, the remaining states only contain a finite number of monomers and may be neglected in the limit. If infimum of internal energies over the remaining states is also equal to $\eps_*$ and they also contain an infinite number of monomers, one can analyze them in a way similar to the previous cases. This scenario is discussed in greater detail in Section~\ref{sec:discrete_example}.

\section{Formal statements and technical calculations}\label{sec:calculations}
Introduce the scaled size distribution functions,
\begin{equation}\label{eq:scale_sdf}
	F_\mu(x;\vv\pn)
	\,=\,
	\frac{f(x/\mu;\vv\pn)}{\mu\E\Mon}
	\,=\,
	\frac{1}{\mu\E\Mon}\sum_{k\geq x/\mu}\pn_k.
\end{equation}
This scaling may be interpreted as setting the widths of the Young diagram cells equal to $\mu$ and the heights --- to $1/(\mu\E\Mon)$; it is unique up to a constant (independent of $\mu$) factor. The reason for this particular scaling becomes apparent from the calculations in Section~\ref{eq:more_calc}. The limit shape (if it exists) is given by the $\mu\dto0$ limit of the expectations of $F_\mu(x)$:
\begin{equation}\label{eq:lim_shape_def}
	F(x)\,\ass\,
	\lim_{\mu\dto0}\E F_\mu(x;\vv\pn)
	\,=\,
	\lim_{\mu\dto0}\;\frac{1}{\mu\E\Mon}\sum_{k\geq x/\mu}\E\pn_k.
\end{equation}
Once we compute the sums in equation \eqref{eq:lim_shape_def} and establish that the $\mu\to0$ limit exists, in order to prove that $F(x)$ is indeed the limit shape, we must also show that the fluctuations around this mean disappear. All of this is summarized in the following theorem:


\begin{theorem}[Limit shapes for partitions of integers]\label{thm:main}
Let the internal energies of classical polymers in a grand-canonical ensemble satisfy Assumption~\ref{ass:nrg_quantum} and the ground state be unattainable; see scenario \ref{case3},  p.~\pageref{ass:nrg_quantum}. Then the function $F(x)$, defined in equation \eqref{eq:lim_shape_def}, exists for all $\beta\geq0$ if $E_k$-s grow sub-logarithmically, and for $\beta\in[0,1)$ if $E_k\sim\ln k$, as $k\to\infty$. Moreover $F(x)$ is the limit shape function, i.e., for every $y>0$ and $\epsilon>0$,
\begin{equation}
	\lim_{\mu\dto0}\;\Prob_\mu\bigg\{\sup_{x\,\geq\,y}\big|F_\mu(x;\vv\pn)\,-\,F(x)\big|\geq\epsilon\bigg\}
	\,=\,0.
\end{equation}
Specific functional expressions for the limit shapes and the corresponding Young diagram scalings are given in Section~\ref{sseq:inf_results}.
\end{theorem}

\begin{proof}
We split the proof into three steps; each step involves a technical calculation which we carry out in the subsequent sections.

\paragraph{1) Establish existence of the thermodynamic limit:} verify whether the expected number of monomers in the system (or the size of a typical Young diagram) tends to infinity as $\mu$ tends to zero. The number of polymers of size $k$, $\vv\pn_k$, is distributed according to equation \eqref{eq:independent}; thus we have,
\begin{equation}\label{eq:exp_mon}
	\E\Mon(\vv\pn)\,=\,
	\sum_{k=1}^\infty k\E\pn_k\,=\,
	\sum_{k=1}^\infty\frac{k\theta_k}{1-\theta_k}\,=\,
	\sum_{k=1}^\infty\frac{k}{\me^{\,\beta E_k+\mu k}-\,1}.
\end{equation}
\noindent The asymptotic behavior of this quantity depends on the behavior of $E_k$-s as $k\to\infty$; it is summarized in Propositions~\ref{prop:therm} and \ref{prop:asympt_mon}, Section~\ref{eq:more_calc}. The conclusion is that the thermodynamic limit may only be attained if the energies grow at most logarithmically as $k\to\infty$.

\paragraph{2) Establish the average behavior of the scaled size distribution functions} defined in equation \eqref{eq:scale_sdf}. 
The limit shape (if exists) is given by the limit of expectation of $F_\mu(x)$, formula \eqref{eq:lim_shape_def}. This calculation is carried out in Section~\ref{sec:asympt_exp} and leads to the formulas for $F(x)$ stated informally in Section~\ref{sseq:inf_results}.

\paragraph{3) Establish that the probability of deviations from the mean tends to zero.} This is done using a Kolmogorov-type inequality to bound the supremum norm of deviations from the mean by the variance of the scaled size distribution functions.

First of all, observe that for any $y>0$, as $\mu\dto0$, $\E F_\mu(x)$ converges to $F(x)$ uniformly on $[y,\infty)$. Indeed, otherwise there would exist some $\epsilon>0$, and sequences $\{\mu_k\}$ (converging to 0) and $\{x_k\}$, such that $|F(x_k)-\E F_{\mu_k}(x_k)|>\epsilon.$ This, however is not possible, because $x_k$ may neither tend to infinity, which would contradict integrability of $F(x)$ and $\E F_\mu(x)$, nor have a converging subsequence, which would contradict pointwise convergence of $\E F_\mu(x)$ to $F(x)$ proven in Proposition~\ref{prop:asympt_exp}.

Denote the supremum norm on $[y,\infty)$ by $\|\cdot\|$. For sufficiently small $\mu$, such that $\|F-\E F_\mu\|\leq \epsilon/2$, we have,
\begin{equation}
	\|F_\mu-F\,\|\,\leq\,\|F_\mu-\E F_\mu\|\,+\,\|F-\E F_\mu\|
	\,\leq\,\|F_\mu-\E F_\mu\|\,+\,\epsilon/2.
\end{equation}
Therefore,
\begin{equation}\label{eq:prob_bound}
	\Prob\big\{\|F_\mu-F\,\|\geq\epsilon\big\}
	\,\leq\,
	\Prob\big\{\|F_\mu-\E F_\mu\|\geq\epsilon/2\big\}.
\end{equation}
At this point we use that in the grand canonical ensembles the size distribution functions are sums of independent random variables, and use a version of Kolmogorov inequality (Lemma~\ref{lem:Kolm} in the Appendix) to bound the expression on the right in equation \eqref{eq:prob_bound}:
\begin{equation}
	\Prob\big\{\|F_\mu-\E F_\mu\|\geq\epsilon/2\big\}
	\,\leq\,
	\frac{4}{\epsilon^2}\V F_\mu(y).
\end{equation}
Proposition~\ref{prop:asympt_var} in Section~\ref{sec:asympt_var} summarizes conditions required for the variance to tend to zero as $\mu\dto0$, and thus the limit shape statement is proven.
\end{proof}

\subsection{Scenario~\ref{case2}: condensation}\label{sec:discrete_example}
Theorem~\ref{thm:main} does not cover scenario~\ref{case2} discussed on p.~\pageref{case2}, i.e., when infimum in equation \eqref{eq:prt_fun_fin} is attained at some finite value(s) of $k$. The thermodynamic limit is then obtained as $\mu\dto\mu_*$, where $\eps_*$ and $\mu_*$ are defined in equations \eqref{eq:prt_fun_fin} and \eqref{eq:mustar} respectively. A {\em condensate} is formed at all states for which $E_k=k\eps_*$, each such state contributing
\begin{equation}\label{eq:finitek}
	\E k\pn_k\,\sim\,\frac{1}{\mu-\mu_*}
\end{equation}
monomers into the total expected number of monomers \eqref{eq:exp_mon}.  As earlier, without loss of generality, we may assume that $\eps_*=0$; this implies that $\mu_*=0$ and that the thermodynamic limit occurs as $\mu\dto0$.

Recall that $\pn_k$-s are independent geometric random variables, see formula \eqref{eq:independent}. This implies that whenever \mbox{$E_k=0$}, their appropriately rescaled versions converge (in distribution), as $\mu\dto0$, to independent exponential random variables:
\begin{equation}\label{eq:exp_conds}
	\Prob\big\{x<\mu\pn_k<y\big\}\,=\,
	(1-\theta_k)\sum_{x/\mu<N<y/\mu}\theta_k^N
	\,\sim\,
	\mu k\sum_{x<\mu N<y}\me^{-\mu N k}
	\,\sim\,k\int_x^y\me^{-kz}\md z.
\end{equation}
If there are only a finite number of condensate states and $\liminf_{k\to\infty}\eps_k>0$, then after this rescaling, the condensate states are the only states which survive  as $\mu\dto0$: $\mu\pn_k\to 0$ (in distribution) for all $k$ such that $E_k>0$. If there is an infinite number of condensate states, then further analysis is required. The limit behavior depends on the density of such states at infinity. We do not carry out this analysis here.

An interesting situation arises if the infimum in \eqref{eq:prt_fun_fin} is attained simultaneously at some finite $k$ and in the limit as $k$ tends to infinity. Then, as the condensate forms at some finite $k$-s, there is also accumulation of monomers at infinity. Applying Proposition~\ref{prop:asympt_mon} to states for which $E_k>0$, we can see that if $E_k=u(\ln k)\ll\ln k$ or $E_k\sim\ln k$ and $\beta<1$, the number of monomers accumulating at infinity grows faster than $1/\mu$, i.e., it dominates the number of monomers accumulating at finite $k$-s: the latter is proportional to $1/\mu$, as in formula~\eqref{eq:finitek}. In this case the limit shape scaling is required and the condensate is negligible in the limit. The situation is opposite if $E_k\sim\ln k$ and $\beta>1$, or $E_k\gg\ln k$:  the number of monomers accumulating at infinity grows slower than $1/\mu$, and therefore the condensate dominates, while the other states become negligible in the limit.

A more delicate analysis is needed if $E_k\sim\ln k$ and $\beta=1$. Let us consider an illustrative case when
\begin{equation}
	E_k\,=\,u(\ln k)\,=\,\ln k;\qquad\theta_k\,=\,\me^{-\mu k}\!/k.
\end{equation} 
Proposition~\ref{prop:asympt_mon} cannot be utilized unless $k=1$ is excluded, as in this case $E_1=0$, which violates Assumption~\ref{ass:nrg_quantum}. 
We can use the same reasoning, however, to get that total number of monomers accumulating at the states with $k>1$ is exactly the same, $1/\mu$, as the number accumulating at the $k=1$ state. Therefore, if we use scaling as in the formula \eqref{eq:exp_conds} above, the mass of the condensate converges to an exponential random variable, while the remaining states vanish as $\mu\dto0$. If, however, we use the limit shape scaling, keep $p_k$-s unaltered, and instead rescale $k$-s, then we get a limit process as described in Remark~\ref{rem:process} on p.~\pageref{rem:process}. In the limit as $\mu\dto0$, $F_\mu(x;\vv\pn)$ remain random functions converging in distribution to an inhomogeneous (backward) Poisson process. The exponential random variable corresponding to $k=1$ should be added as the atomic component at $x=0$.


\subsection{Asymptotics for the expected number of monomers}\label{eq:more_calc}
\begin{prop}[Rough asymptotics of $\E\Mon$]\label{prop:therm}
		Consider a sequence of (strictly) positive numbers, $E_k$; $k\in\N$. Depending on the asymptotic behavior of $E_k$-s as $k\to\infty$, the expected number of monomers given by formula \eqref{eq:exp_mon} has the following limits as $\mu\dto0$:
		\begin{enumerate}[label=(\roman*)]
			\item $\boldsymbol{E_k\gg\ln k}$: $\E\Mon$ is bounded for all $\beta>0$;
			\item $\boldsymbol{E_k\sim\ln k}$: $\E\Mon$ is bounded for all $\beta>2$ and tends to infinity for $\beta<2$;
			\item $\boldsymbol{E_k\ll\ln k}$: $\E\Mon\to\infty$ for all $\beta>0$.
		\end{enumerate}
\end{prop}
\begin{proof}
%
%
%
%
Consider the supercritical case (i). There exists some $K>1$, such that $\beta E_k>3\ln k$ for all $k>K$; for such $k$-s we have a bound,
\begin{equation}
	\me^{\,\beta E_k+\mu k}-\,1\,>\,
	k^3\me^{\,\mu k}-\,1\,>\, 
	k^3/\,2.
\end{equation}
Using this bound in equation \eqref{eq:exp_mon}, we get that for all $\mu>0$,
\begin{equation}\label{eq:up_bnd1}
	\E\Mon(\vv\pn)\,<\,
	\sum_{k=1}^{K-1}\frac{k}{\me^{\,\beta E_k}-\,1}\,+\,2\sum_{k=K}^\infty\frac{1}{k^2}\,<\,C.
\end{equation}
The constant $C$ does not depend on $\mu$, thus $\E\Mon$ remains bounded as $\mu\dto0$.

Consider the subcritical case (iii).  Here we can find $K$ such that for all $k>K$, $\beta E_k<2\ln k$, and we get a bound,
\begin{equation}
	\me^{\,\beta E_k+\mu k}-\,1\,<\,
	k^2\me^{\,\mu k}.
\end{equation}
This implies that
\begin{equation}
	\E\Mon(\vv\pn)\,>\,
	\sum_{k=K}^\infty\frac{\me^{-\mu k}}{k}\,\to\,\infty
	\qquad\text{as}\qquad\mu\dto0.
\end{equation}

Now consider the critical case (ii). We have, $E_k=\ln k+\delta_k$, where $\delta_k\ll\ln k$ as $k\to\infty$. Suppose $\beta>2$; then we can pick some $\alpha\in(2,\beta)$ and $K>1$, such that for all $k>K$, $\beta\delta_k>(2-\alpha)\ln k$. Therefore,
\begin{equation}
	\me^{\,\beta E_k+\mu k}-\,1\,>\,
	k^{\beta-\alpha+2}\me^{\,\mu k}-\,1\,>\, 
	k^{\beta-\alpha+2}/\,2;
\end{equation}
and we get a bound similar to \eqref{eq:up_bnd1}:
\begin{equation}
	\E\Mon(\vv\pn)\,<\,
	\sum_{k=1}^{K-1}\frac{k}{\me^{\,\beta E_k}-\,1}\,+\,2\sum_{k=K}^\infty\frac{1}{k^{1+\beta-\alpha}}\,<\,C.
\end{equation}
If $\beta<2$, we can pick $\alpha\in(\beta,2)$ and a large enough $K$, such that for all $k>K$, $\beta\delta_k<(2-\alpha)\ln k$. Then
\begin{equation}
	\me^{\,\beta E_k+\mu k}-\,1\,<\,
	k^{\beta-\alpha+2}\me^{\,\mu k};
\end{equation}
and we get,
\begin{equation}
	\E\Mon(\vv\pn)\,>\,
	\sum_{k=K}^\infty\frac{\me^{-\mu k}}{k^{1+\beta-\alpha}}\,\to\,\infty
	\qquad\text{as}\qquad\mu\dto0.
\end{equation}

Note, that if $\beta=2$, various behaviors are possible depending on the asymptotics of $\delta_k$-s, e.g., if $\delta_k\equiv0$, $\E\Mon$ diverges as $\mu\dto0$, whereas if $\delta_k=\ln\ln k$, $\E\Mon$ remains bounded. As in the classical problem regarding convergence of series, there exists no ``borderline'' asymptotic behavior of $\delta_k$-s which would separate convergent and divergent behaviors of $\E\Mon$.
\end{proof}

\begin{rem} In the cases when $\E\Mon$ remains bounded as $\mu\dto0$, the thermodynamic limit cannot be achieved in the setting of grand canonical ensembles. One must study the canonical ensembles directly (we do not carry out this study here). 
\end{rem}

We can make a stronger statement regarding the asymptotics of $\E\Mon$ if we impose additional restrictions on $E_k$-s. Here is a refinement of Proposition~\ref{prop:therm} for energies given by formula \eqref{eq:energies_u}:

\begin{prop}[Fine asymptotics of $\E\Mon$]\label{prop:asympt_mon}
	Suppose the polymer energies satisfy Assumption~\ref{ass:nrg_quantum} (p.~\pageref{ass:nrg_quantum}). Then the expected number of monomers has the following asymptotic behavior as $\mu\dto0$:
\begin{equation}\label{eq:fine_ass_mon}
	\E\Mon(\vv\pn)
	\,\sim\,
	\lambda\,\mu^{-2}\me^{-\beta u(-\ln\mu)};\qquad
	\lambda\,=\,\int_0^\infty x\,\Phi(x)\md x.
\end{equation}
The function $\Phi(\cdot)$ is determined by the asymptotic behavior of $u(x)$ as $x\to\infty$, as presented in Table~\ref{tab:scalings}. An additional condition, $\beta<2$, is required in case (iv): according to Proposition~\ref{prop:therm}, $\E\Mon$ remains finite if $\beta>2$.
\end{prop}

\begin{table}[t] 
\begin{center}
\begin{tabular}{rccclll}
	&\rule{0.25em}{0pt}&		$\displaystyle\lim_{x\to\infty} u(x)$	&	$\displaystyle\lim_{x\to\infty} u^\prime(x)$	&	Typical $E_k=u(\ln k)$		&	$\Phi(x)$	&$\lambda$\\
	\cmidrule[\heavyrulewidth]{3-7}
	i&&		$0$		&	$0$			&	$1/\,k^\alpha$,\quad$\alpha>0$	& 	$1/(\me^x-1)$				&	$\pi^2\!/\,6$\\	
	ii&&		$1$		&	$0$			&	$1$		&	$1/(\me^{x}-\me^{-\beta})$			&	$\me^{\,\beta}\Li_2\big(\me^{-\beta}\big)$\\
	iii&&		$\infty$	&	$0$			&	$\ln\ln k$	& 	$\me^{-x}$						&	$1$\\			
	iv&&		$\infty$	&	$1$			&	$\ln k$		&	$x^{-\beta}\me^{-x}$				&	$\Gfun(2-\beta)$\\
	\cmidrule[\lightrulewidth]{3-7}
\end{tabular}
\end{center}
\caption{\label{tab:scalings}
	Possible asymptotic behaviors of the function $u(\cdot)$ prescribing internal energies of polymers via relation \eqref{eq:energies_u}. The function $\Phi(x)$ and $\lambda$ are related to the asymptotic behavior of the expected number of monomers in the system as asserted in Proposition~\ref{prop:asympt_mon}.  If $u^\prime(x)$ tends to infinity as $x\to\infty$, according to Proposition~\ref{prop:therm}, $\E\Mon$ remains bounded as $\mu\dto0$, i.e., there is no thermodynamic limit.
}
\end{table}

\begin{proof}
Substituting $E_k=u(\ln k)$ into equation \eqref{eq:exp_mon}, we get,
\begin{equation}
	\E\Mon(\vv\pn)\,=\,
	\sum_{k=1}^\infty\frac{k}{\me^{\,\beta u(\ln k)+\mu k}-\,1}
	\,=\,
	\frac{1}{\mu^2\me^{\,\beta u(-\ln\mu)}}\sum_{k=1}^\infty\frac{\mu k}{\me^{\,\beta[u(\ln\mu k-\ln\mu)-u(-\ln\mu)]+\mu k}\,-\,\me^{-\beta u(-\ln\mu)}}\,\mu.
\end{equation}
The sum in the formula above may be transformed into an integral, yielding
\begin{equation}
	\mu^2\me^{\,\beta u(-\ln\mu)}\E\Mon(\vv\pn)\,=\,
	\int_0^\infty\Psi_\mu(x)\md x,
\end{equation}
where the function $\Psi_\mu(x)$ is piece-wise constant with values
\begin{equation}
	\Psi_\mu(x)\,=\,
	\frac{\mu k}{\me^{\,\beta[u(\ln \mu k-\ln\mu)-u(-\ln\mu)]+\mu k}\,-\,\me^{-\beta u(-\ln\mu)}}
\end{equation}
for $x\in[\mu(k-1);\mu k)$; $k\in\N$. Fix some $x>0$; as $\mu\dto0$, $k$ in the formula above is selected so that $|\mu k-x|\leq\mu\to0$, i.e., $\mu k\to x$. Using Lemma~\ref{lem:infinity} from the Appendix, we get that for any given $x$, 
\begin{equation}
	\lim_{\mu\dto0}\big[u(\ln \mu k-\ln\mu)-u(-\ln\mu)\big]\,=\,\lim_{t\to\infty} u^\prime(t)\,\ln x
\end{equation}
Thus, as $\mu\dto0$, $\Psi_\mu(x)$ converges pointwise to $x\Phi(x)$, where $\Phi(x)$ is one of the functions presented in Table~\ref{tab:scalings}; it depends on the behavior of $u(\cdot)$ and $u^\prime(\cdot)$ at infinity. It is straightforward to verify that $\Psi_\mu(x)<C\,x\Phi(x)$ for all small enough $\mu$ and a suitably chosen constant $C$; therefore, by the dominated convergence theorem we get the desired result. The condition $\beta<2$ is needed in case (iv) so that $x\Phi(x)$ is integrable near zero.
\end{proof}


\subsection{Asymptotics of the expected value of the size distribution function}\label{sec:asympt_exp}
\begin{prop}\label{prop:asympt_exp}
	Suppose the polymer energies satisfy Assumption~\ref{ass:nrg_quantum} (p.~\pageref{ass:nrg_quantum}). Then for any $x>0$, the expectation of the scaled size distribution function \eqref{eq:scale_sdf} has the following asymptotic behavior as $\mu\dto0$:
	\begin{equation}
		F(x)\,\ass\,
		\lim_{\mu\dto0}\E F_\mu(x;\vv\pn)\,=\,
		\frac{1}{\lambda}\int_x^\infty\,\Phi(y)\md y.
	\end{equation}
	The function $\Phi(\cdot)$ and the constant $\lambda$ are as presented in Table~\ref{tab:scalings}.
\end{prop}

\begin{proof}
Recall that $\pn_k$-s are geometric random variables with parameters $\theta_k$ \eqref{eq:independent}, thus we have,
\begin{equation}\label{eq:exp_f}
	\sum_{k\geq x/\mu}\E\pn_k\,=\,
	\sum_{k\geq x/\mu}\frac{\theta_k}{1-\theta_k}\,=\,
	\sum_{k\geq x/\mu}\frac{1}{\me^{\,\beta E_k+\mu k}-\,1}.
\end{equation}
After some rearrangement, using that $E_k=u(\ln k)$, we continue:
\begin{equation}\label{eq:some1}
	\cdots
	\,=\,
	\frac{1}{\mu\me^{\,\beta u(-\ln\mu)}}\sum_{k\geq x/\mu}\frac{1}{\me^{\,\beta[u(\ln\mu k-\ln\mu)-u(-\ln\mu)]+\mu k}\,-\,\me^{-\beta u(-\ln\mu)}}\,\mu
	\,=\,
	\frac{1}{\mu\me^{\,\beta u(-\ln\mu)}}\int_{\underaccent{\tilde}{x}}^\infty\Phi_\mu(y)\md y,
\end{equation}
where $\underaccent{\tilde}{x}=\mu\floor{x/\mu}$ (the {\em floor} of $x$, $\floor{x}$, is the greatest integer smaller than or equal to $x$) and the function $\Phi_\mu(\cdot)$ is defined similarly to how we defined $\Psi_\mu(\cdot)$ in the proof of Proposition~\ref{prop:asympt_mon}, p.~\pageref{prop:asympt_mon}, except without the factor of $\mu k$. Combining formula \eqref{eq:some1} with asymptotics for $\E\Mon$ \eqref{eq:fine_ass_mon}, we get that as $\mu\dto0$,
\begin{equation}
	\E F_\mu(x;\vv\pn)\,=\,\frac{1}{\mu\E\Mon}\sum_{k\geq x/\mu}\E\pn_k\,\sim\,
	\frac{1}{\lambda}\int_0^\infty\,\id_{(\underaccent{\tilde}{x},\infty)}(y)\Phi_\mu(y)\md y.
\end{equation}
The function $\id_{(\underaccent{\tilde}{x},\infty)}(y)\Phi_\mu(y)$ converges pointwise to $\id_{(x,\infty)}(y)\Phi(y)$ and is bounded by $C\,\Phi(y)$ for a suitably chosen constant $C$, so as earlier, we get the desired result by the dominated convergence theorem. Note that $x$ may also be set to 0 in cases (ii, iii), and (iv) when $\beta<1$, see Table~\ref{tab:scalings}, as then $\Phi(\cdot)$ is integrable over $[0,\infty)$.
\end{proof}

%

\subsection{Asymptotics of the variance of the size distribution function} \label{sec:asympt_var}
\begin{prop}\label{prop:asympt_var}
	Suppose the polymer energies satisfy Assumption~\ref{ass:nrg_quantum} (p.~\pageref{ass:nrg_quantum}). Then, depending on the asymptotic behavior of $u(\cdot)$ at infinity, see Table~\ref{tab:scalings}, for any $x>0$, the variance of the scaled size distribution function \eqref{eq:scale_sdf} has the following asymptotic behaviors as $\mu\dto0$:
	\begin{equation}
		\lim_{\mu\dto0}\V F_\mu(x;\vv\pn)\,=\,0\quad\text{cases (i-iii) and (iv) when }\beta<1;\qquad
		\lim_{\mu\dto0}\V F_\mu(x;\vv\pn)\,=\,\infty\quad\text{case (iv) when }\beta>1.
	\end{equation}
\end{prop}
\begin{proof}
Recall that $p_k$-s are independent geometric random variables with parameters $\theta_k$ \eqref{eq:independent}, thus
\begin{align}\label{eq:var_f}
	\V\sum_{k\geq x/\mu}\pn_k\,=\,
	\sum_{k\geq x/\mu}\V\pn_k\,=\,
	\sum_{k\geq x/\mu}\frac{\theta_k}{(1-\theta_k)^2}\,&=\,
	\sum_{k\geq x/\mu}\frac{1}{\me^{\,\beta E_k+\mu k}-\,2\,+\,\me^{-\beta E_k-\mu k}}.
\end{align}
Proceeding in the same way as in the proof of Proposition~\ref{prop:asympt_exp}, we get
\begin{equation}
	\cdots
	\,=\,
	\frac{1}{\mu\me^{\,\beta u(-\ln\mu)}}\sum_{k\geq x/\mu}\frac{1}{\me^{\,\beta[u(\ln\mu k-\ln\mu)-u(-\ln\mu)]+\mu k}\,-\,2\me^{-\beta u(-\ln\mu)}\,+\,\me^{-\beta[u(\ln k)+u(-\ln\mu)]-\mu k}}\,\mu,
\end{equation}
which may be converted into an integral of a suitable piecewise constant function, $\Upsilon_\mu(\cdot)$. Skipping the details, which are identical to our previous derivations, we get,
\begin{equation}\label{eq:var_scaling}
	\V F_\mu(x;\vv\pn)\,=\,\frac{1}{(\mu\E\Mon)^2}\V\sum_{k\geq x/\mu}\pn_k\,\sim\,
	\frac{\mu\me^{\,\beta u(-\ln\mu)}}{\lambda^2}\int_0^\infty\,\id_{(\underaccent{\tilde}{x},\infty)}(y)\Upsilon_\mu(y)\md y.
\end{equation}
Depending on the asymptotics of $u(\cdot)$ and $u^\prime(\cdot)$ at infinity (as in Table~\ref{tab:scalings}), $\Upsilon_\mu(y)$ converges pointwise to one of the following functions:
\begin{equation}
	\text{(i)}\quad\frac{1}{2(\cosh y-1)};\qquad\text{(ii)}\quad
	\frac{\me^{\,\beta}}{2\big(\cosh(y+\beta)-1\big)};\qquad\text{(iii)}\quad\me^{-y};\qquad\text{(iv)}\quad y^{-\beta}\me^{-y}.
\end{equation}
Thus the integral in formula \eqref{eq:var_scaling} is finite as long as $x>0$; moreover, $x=0$ may also be included in cases (ii, iii) and in case (iv) as long as $\beta<1$. Variance of the scaled size distribution functions tends to 0 when $\mu\dto0$, whenever \smash{$\mu\me^{\,\beta u(-\ln\mu)}$} does. This is the case if $u(\cdot)$ is finite or grows sublinearly at infinity. In case of linear growth, as long as $\beta<1$, the variance converges to zero; if $\beta>1$ it tends to infinity.
\end{proof}
\begin{rem}\label{rem:process}
In the proof of Proposition~\ref{prop:asympt_var} above, one can see that if $u(x)\sim x$ as $x\to\infty$ and \mbox{$\beta=1$}, the variance of the scaled size distribution function has a nonzero limit as $\mu\dto0$. In this case there is no limit shape and the scaled size distribution functions remain random in the limit. A detailed analysis of this scenario could be a subject of a different study, but let us sketch a calculation for one particular case when $u(x)=x$, i.e., $E_k=\ln k$. Thus $\pn_k$-s are independent geometric random variables with parameters \smash{$\theta_k=\me^{-\mu k}\!/k$}. Proposition~\ref{prop:asympt_mon} cannot be used for $k=1$, as in this case $u(\ln 1)=0$, which violates Assumption~\ref{ass:nrg_quantum}. In fact, there is {\em condensation} at $k=1$, see Section~\ref{sec:discrete_example}. We can still use Proposition~\ref{prop:asympt_mon} for $k>1$ to get that as $\mu\dto0$ (assuming that the $k=1$ state is excluded)
\begin{equation}
	\E\Mon(\vv\pn)\,\sim\,\frac{1}{\mu};\qquad
	F_\mu(x;\vv\pn)\,\sim\,\sum_{k\geq x/\mu}\pn_k.
\end{equation}
Observe that the scaled size distribution functions remain integer-valued in this limit. Let us calculate the characteristic function of the difference, $F_\mu(x;\vv\pn)-F_\mu(y;\vv\pn)$ for some $x$ and $y$ such that $0<x<y$:
\begin{equation}
	\phi(x,y;t)\,\ass\,
	\E\exp\left\{\mi t\sum_{x/\mu\leq k<y/\mu}\pn_k\right\}\,=\,
	\prod_{x/\mu\leq k< y/\mu}\E\me^{\,\mi t\pn_k}\,=\,
	\prod_{x/\mu\leq k< y/\mu}\frac{1-\me^{-\mu k}\!/k}{1-\me^{\,\mi t -\mu k}\!/k}.
\end{equation}
The logarithm of $\phi(x,y;t)$ may be converted into an integral as in the proofs of Propositions~\ref{prop:asympt_mon}-\ref{prop:asympt_var}, and we can get that as $\mu\dto0$,
\begin{equation}
	\ln\phi(x,y;t)\,\sim\,
	\lambda(x,y)
	\left(
	\me^{\,\mi t}\,-\,1
	\right);\qquad
	\lambda(x,y)\,=\,\int_x^y\frac{\me^{-z}}{z}\md z.
\end{equation}
This is the characteristic function of a Poisson process with parameter $\lambda(x,y)$. Therefore, as $\mu\dto0$, the scaled size distribution functions converge in distribution to a (backward in $x$) Poisson process which starts at zero when ``$x=\infty$,'' tends to infinity (almost surely) when $x\dto0$, and whose jumps are distributed with density $\me^{-x}\!/x$. 
\end{rem}

\section{Acknowledgement}
The authors are grateful to the referees for valuable suggestions and references. I.F. acknowledges support by the NSF grant DMS-1056471. V.S. acknowledges support by the Leverhulme Research Grant, RPG-2014-226.


\newpage

\section{Technical lemmas}\label{sec:lemmas}

\begin{lemma}[A version of Kolmogorov's inequality]\label{lem:Kolm}
	Let $x_n$, $n\in\N$ be independent random variables, such that
	\begin{equation}
		\sum_{n=1}^\infty\E x_n<\infty,
		\qquad\qquad
		\sum_{n=1}^\infty\V x_n<\infty.
	\end{equation}
	Set $S_N\ass\sum_{n=N}^\infty x_n$. Then for any $N$, $S_N$ exists almost surely and
	\begin{equation}
		\Prob\bigg\{\sup_{n\geq N}|S_n-\E S_n|\geq\epsilon\bigg\}
		\,\leq\,
		\frac{1}{\epsilon^2}\sum_{n=N}^\infty\V x_n
		\,=\,
		\frac{\V S_N}{\epsilon^2}.
	\end{equation}
\end{lemma}

\begin{proof}
The almost sure existence of $S_N$ is guaranteed by the Two-Series theorem \cite{shiryaev1996probability}. To simplify notation, introduce the centered variables, $\bar x_n\ass x_n-\E x_n$; $\bar S_n\ass S_n-\E S_n$. Introduce events that $|\bar S_n|$ exceeds $\epsilon$ for some $n<M$, i.e.,  a sequence of increasing sets,
\begin{equation}
	\set A_M\,\ass\,\bigg\{\max_{N\leq n<M}|\bar S_n|\geq\epsilon\bigg\}
	\,=\,\bigg\{\max_{N\leq n<M}\bigg|\sum_{k=n}^{M-1}\bar x_k\,+\,\bar S_M\bigg|\geq\epsilon\bigg\}.
\end{equation}
The supremum of $|\bar S_n|$ is either reached at some finite $n$, or in the limit as $n$ tends to infinity, i.e.,
\begin{equation}
	\bigg\{\sup_{n\geq N}|\bar S_n|\geq\epsilon\bigg\}
	\,=\,
	\bigg\{\limsup_{n\to\infty}|\bar S_n|\geq\epsilon\bigg\}
	\cup\bigcup_{M>N}\set A_M.
\end{equation}
As $\lim_{n\to\infty}\bar S_n=0$ almost surely, $\Prob\Big\{\limsup_{n\to\infty}|\bar S_n|\geq\epsilon\Big\}=0$; therefore,
\begin{equation}
	\Prob\bigg\{\sup_{n\geq N}|\bar S_n|\geq\epsilon\bigg\}
	\,=\,
	\Prob\bigcup_{M>N}\set A_M
	\,=\,
	\lim_{M\to\infty}\Prob\set A_M.
\end{equation}
Thus we get,
\begin{equation}
	\Prob\bigg\{\sup_{n\geq N}|\bar S_n|\geq\epsilon\bigg\}
	\,=\,
	\lim_{M\to\infty}\Prob\bigg\{\max_{N\leq n<M}\bigg|\sum_{k=n}^{M-1}\bar x_k\,+\,\bar S_M\bigg|\geq\epsilon\bigg\}.
\end{equation}
Now we can use the standard Kolmogorov inequality to obtain the desired result:
\begin{equation}
	\Prob\bigg\{\sup_{n\geq N}|\bar S_n|\geq\epsilon\bigg\}
	\,\leq\,\limsup_{M\to\infty}
	\frac{1}{\epsilon^2}
	\bigg(\sum_{n=N}^{M-1}\V\bar x_n\,+\,\V\bar S_M\bigg)
	\,=\,\frac{1}{\epsilon^2}\sum_{n=N}^\infty\V \bar x_n
	\,=\,
	\frac{\V \bar S_N}{\epsilon^2}.
\end{equation}
\end{proof}

\begin{lemma}\label{lem:infinity}
Let $u(t)$ be a differentiable function such that\, $\lim_{t\to\infty}u^\prime(t)=\sigma$; let $\delta(t)\to\Delta$ as $t\to\infty$. Then
	\begin{equation}\label{eq:lim_fun}
		\lim_{t\to\infty}\big[u(t+\delta(t))-u(t)\big]\,=\,\sigma\Delta.
	\end{equation}
\end{lemma}

\begin{proof}
By virtue of the mean value theorem, there exists some $s(t)\in[t,t+\delta(t)]$, such that
\begin{equation}
	u(t+\delta(t))\,-\,u(t)
	\,=\,
	u^\prime(s(t))\,\delta(t).
\end{equation}
When $t\to\infty$, $s(t)\to\infty$ as well, thus passing to the limit, we obtain the desired result.
\end{proof}
%
%


\bibliographystyle{plain} 
\bibliography{../bibliography/bibl}

\end{document}